\documentclass{ws-ijgmmp}
\def\SU{{\rm SU}}
\def\C{{\mathbb C}}
\def\diag{{\rm diag}}
\def\R{{\mathbb R}}
\def\Z{{\mathbb Z}}
\def\H{{\mathbb H}}

\def\U{{\rm U}}

\def\SO{{\rm SO}}
\def\GL{{\rm GL}}

\def\OO{{\rm O}}
\def\T{{\rm T}}
\def\HH{{\rm H}}
\def\Spin{{\rm Spin}}
\def\cl{{C}\!\ell}
\def\Mat{{\rm Mat}}
\def\cen{{\rm cen}}

\allowdisplaybreaks

\begin{document}

\markboth{Dmitry Shirokov}
{Calculation of Spin Group Elements Revisited}

%
\catchline{}{}{}{}{}
%

\title
{CALCULATION OF SPIN GROUP ELEMENTS REVISITED}

\author{DMITRY SHIROKOV}

\address{HSE University\\
Moscow, 101000, Russia\\
\email{dshirokov@hse.ru} }

\address{Institute for Information Transmission
Problems of Russian Academy of Sciences\\
Moscow, 127051, Russia\\
shirokov@iitp.ru }

\maketitle

\begin{history}
\received{(15 October 2024)}
\revised{(Day Month Year)}
\end{history}

\begin{abstract}
In this paper, we present a method for calculation of spin groups elements for known pseudo-orthogonal group elements with respect to the corresponding two-sheeted coverings. We present our results using the Clifford algebra formalism in the case of arbitrary dimension and signature, and then explicitly using matrices, quaternions, and split-quaternions in the cases of all possible signatures $(p,q)$ of space up to dimension $n=p+q=3$. The different formalisms are convenient for different possible applications in physics, engineering, and computer science.

\end{abstract}


\section{Introduction}	

Two-sheeted coverings of (pseudo-)orthogonal groups by spin groups are widely used in various applications. The relation between a fixed (pseudo-)orthogonal matrix and the corresponding pair of spin group elements is nonlinear. The problem arises of how to calculate, for a known (pseudo-)orthogonal matrix, the corresponding pair of spin group elements. The natural language describing this relation is the Clifford algebra formalism, but for many applications the matrix formalism is preferable. In this paper, we give a complete answer to this problem, firstly, using the Clifford algebra formalism in the case of arbitrary dimension and signature and, secondary, using the matrix, quaternion, and split-quaternion formalisms in the cases of all possible signatures $(p,q)$ of space up to dimension $n=p+q=3$. For known methods in the Clifford algebra (or geometric algebra) formalism for the cases $n=3, 4$, we refer the reader to \cite{Hestenes, Doran, Lounesto}. In this paper, we generalize the results of \cite{Shirokovspin, Shirokovspin2} and present a new explicit method that works in the case of arbitrary dimension $n=p+q\geq 1$ and for an arbitrary fixed (pseudo-)orthogonal matrix.

The paper is organized as follows. In Section \ref{sectionClifford}, we present new results on the calculation of spin group elements using the Clifford algebra formalism. In Section \ref{sectionMatrix}, we present new results using the matrix, quaternion, and split-quaternion formalisms. Conclusions follow in Section~\ref{Conclusions}.

\section{Clifford Algebra Formalism}\label{sectionClifford}

Let us consider the real Clifford algebra (or geometric algebra) $\cl_{p,q}$ \cite{Hestenes,Doran,Lounesto,Shirokovspin2} with the identity element $e\equiv 1$ and the generators $e_a$, $a=1, 2, \ldots, n=p+q$. The generators satisfy the relations
\begin{eqnarray}
e_a e_b+e_b e_a=2\eta_{ab}e,\qquad \eta=(\eta_{ab})=\diag(\underbrace{1, \ldots , 1}_p, \underbrace{-1, \ldots, -1}_{q}).\label{eta}
\end{eqnarray}
In Euclidean cases, the notation $\cl_n:=\cl_{n, 0}$ is used. We use ordered multi-indices $A$ of length from $0$ to $n$ to denote $2^n$ fixed basis elements of $\cl_{p,q}$:
\begin{eqnarray}
e_A=e_{a_1 \ldots a_k}=e_{a_1}e_{a_2}\cdots e_{a_k},\qquad a_1 < a_2 < \cdots < a_k.
\end{eqnarray}
In particular, the identity element $e_\emptyset=e$ has an empty multi-index $A=\emptyset$ with zero length $|A|=0$. An arbitrary element (multivector) $U\in\cl_{p,q}$ has the form
\begin{eqnarray}
U&=&ue+\sum_{a=1}^n u_a e_a+\sum_{a<b} u_{ab}e_{ab}+\cdots+u_{1\ldots n}e_{1\ldots n}\\
&=&\sum_{k=0}^n\sum_{a_1<\cdots<a_k}u_{a_1 \ldots a_k} e_{a_1 \ldots a_k},\qquad u_A\in\R.
\end{eqnarray}
The Clifford algebra $\cl_{p,q}$ can be represented as the direct sum
\begin{eqnarray}
\cl_{p,q}=\bigoplus_{k=0}^n \cl^k_{p,q},\quad \cl^k_{p,q}=\{\sum_{a_1<\cdots<a_k}u_{a_1 \ldots a_k} e_{a_1 \ldots a_k}\},\quad \dim (\cl_{p,q}^k)=C_n^k,
\end{eqnarray}
where $\cl_{p,q}^k$ are called subspaces of grades $k=0, 1, \ldots, n$. We denote the operation of projection onto $\cl^k_{p,q}$ by $\langle U \rangle_k$. We need the even and odd subspaces of $\cl_{p,q}$ with the following definitions:
\begin{eqnarray}
\cl^{(0)}_{p,q}:=\bigoplus_{k=0\mod 2} \cl^k_{p,q},\qquad \cl^{(1)}_{p,q}:=\bigoplus_{k=1\mod 2}\cl^k_{p,q}.
\end{eqnarray}
We use the following standard operation of reversion in $\cl_{p,q}$
\begin{eqnarray}
\widetilde{U}:=\sum_{k=0}^n(-1)^{\frac{k(k-1)}{2}} \langle U \rangle_k
\end{eqnarray}
with the property
\begin{eqnarray}
\widetilde{UV}=\widetilde{V}\,\widetilde{U},\qquad \forall U, V\in\cl_{p,q}.
\end{eqnarray}
Let us consider the spin groups
\begin{eqnarray}
\Spin_+(p,q):=\{T\in\cl^{(0)}_{p,q}\cup\cl^{(1)}_{p,q} \,|\, T^{-1} \cl^1_{p,q}T\subset \cl^1_{p,q},\quad \widetilde{T}T= e\}.\label{spin}
\end{eqnarray}
and (pseudo-orthogonal) special orthochronous groups 
\begin{eqnarray}
\SO_+(p,q):=\{P\in\Mat(n,\R)\,|\,P^\T \eta P = \eta,\quad \det(P)=1,\quad p^{1\ldots p}_{1\ldots p}\geq 1\}.
\end{eqnarray}
Here and below we use the notation $p^A_B$ with ordered multi-indices $A=a_1 \ldots a_k$ and $B=b_1 \ldots b_k$ to denote the minor of the matrix $P$ with matrix entries from intersection of rows $a_1$, \ldots, $a_k$ and columns $b_1$, \ldots, $b_k$. For example, $p^{1\ldots p}_{1\ldots p}$ is the minor from the upper left corner of the matrix $P$ of size $p$. In the case of empty multi-indices $A$ and $B$, the corresponding minor $p^\emptyset_\emptyset$ is equal to $1$ by definition. The groups $\SO_+(p,q)$ are connected components of the identity of the pseudo-orthogonal groups 
\begin{eqnarray}
\OO(p,q):=\{P\in\Mat(n,\R)\,|\,P^\T \eta P = \eta\}.
\end{eqnarray}
For the cases of Euclidean signatures, we have $\SO(n)=\SO_+(n,0)$, $\OO(n)=\OO(n,0)$, and $\Spin(n)=\Spin_+(n,0)$.

We have two-sheeted coverings of the special orthochronous groups $\SO_+(p,q)$ by the spin groups $\Spin_+(p,q)$. From the algebraic point of view, this means that for an arbitrary matrix $P=(p_a^b)\in\SO_+(p,q)$ there are two elements $S=\pm T\in\Spin_+(p,q)$ such that 
\begin{eqnarray}
Se_aS^{-1}=p_a^b e_b,\label{conn}
\end{eqnarray}
where we have the sum over index $b$ from $1$ to $n$ (here and below we use the Einstein convention). 

If we know the elements $S=\pm T\in\Spin_+(p,q)$, then it is easy to find the corresponding matrix $P\in\SO_+(p,q)$ from (\ref{conn}). The problem of finding $P\in\SO_+(p,q)$ for known $S=\pm T\in\Spin_+(p,q)$ is more difficult. The following theorem gives an answer to this problem in the case of arbitrary $n=p+q$.

We use the notation 
\begin{eqnarray}
e^A:=(e_A)^{-1}
\end{eqnarray}
for the inverse of basis elements of $\cl_{p,q}$. Note that we use the Einstein convention for ordered multi-indices too. We have the sum over all ordered multi-indices $A$ and $B$ of the same length $|A|=|B|$ in (\ref{MMF}).

Combining the methods from \cite{Shirokovspin} and \cite{Shirokovspin2}, we get the following new result. Note that \cite{Shirokovspin2} presents a particular case of the following theorem (only an explicit formula with $M_\emptyset$ is presented, which does not work in the general case).

\begin{theorem}\label{thnew}
Let $P\in\SO_+(p,q)$, $p+q=n$. We can always choose 
a basis element $e_F\in\{e_D \, | \, |D|=0\mod 2\}\in\cl_{p,q}$ such that
\begin{eqnarray}
M_F:=p^B_A e_B e_F e^A \neq 0.\label{MMF}
\end{eqnarray}
Next, we can find elements $S=\pm T\in\Spin_+(p,q)$ that correspond to $P=(p_a^b)\in\SO_+(p,q)$ as the two-sheeted covering $Se_aS^{-1}=p_a^b e_b$ in the following way:
\begin{eqnarray}
S=\pm \frac{M_F}{\sqrt{\widetilde{M_F}M_F}}.\label{SSF}
\end{eqnarray}
\end{theorem}

\begin{proof} Multiplying (\ref{conn}) by itself several times, we get (see \cite{Shirokovspin2} for details)
\begin{eqnarray}
Se_AS^{-1}=p_A^B e_B.\label{multi}
\end{eqnarray}
Multiplying both sides of (\ref{multi}) on the right by $e_F$ and $e^A:=(e_A)^{-1}$, we get
\begin{eqnarray}
Se_AS^{-1}e_F e^A=p_A^B e_B e_F e^A,\label{tyu}
\end{eqnarray}
where we have the sum over multi-indices $A$ and $B$ of the same length $|A|=|B|$. 

Let us denote the right-hand side of (\ref{tyu}) by $M_F$. The fact that we can always find an element $e_F\in\{e_D \, | \, |D|=0\mod 2\}$ such that $M_F\neq 0$ follows from the results of \cite{Shirokovspin} (see also Theorem 4 in \cite{Dev}).

We have 
$$e_A U e^A=2^n\langle U \rangle_{\cen},$$
where we denote by $\langle U\rangle_{\cen}$ the projection of $U\in\cl_{p,q}$ onto the center $\cen(\cl_{p,q})$ of $\cl_{p,q}$, which is $\cl^0_{p,q}$ in the case of even $n$ and $\cl^0_{p,q}\oplus\cl^n_{p,q}$ in the case of odd $n$. Since $S, e_F\in\cl^{(0)}_{p,q}$, we obtain $\langle S^{-1} e_F \rangle_{\cen}=\langle S^{-1}e_F \rangle_0$. 

From (\ref{tyu}), we get
\begin{eqnarray}
2^n S \langle S^{-1}e_F \rangle_{0}=M_F.\label{g2}
\end{eqnarray}
Taking reversion of both sides (\ref{g2}), we get
\begin{eqnarray}
2^n \langle S^{-1}e_F \rangle_{0} \widetilde{S}=\widetilde{M_F}.\label{g1}
\end{eqnarray}
Multiplying both sides of (\ref{g1}) by both sides of (\ref{g2}), we obtain
\begin{eqnarray}
(2^n)^2  \langle S^{-1}e_F \rangle_{0}  \widetilde{S}S \langle S^{-1}e_F\rangle_{0}=\widetilde{M_F} M_F.
\end{eqnarray}
Taking into account $\widetilde{S}S=e$, we get
\begin{eqnarray}
(2^n  \langle S^{-1}e_F \rangle_{0})^2=  \widetilde{M_F} M_F.
\end{eqnarray}
Taking the square root and substituting this expression into (\ref{g2}), we obtain the statement of the theorem.
\end{proof}

In terms of frames and rotors of geometric algebra, we can reformulate this theorem in the following way.

\begin{corollary} Let us have two frames $e_a$ and $\beta_a$ related by the rotation
$$
Se_a \widetilde{S}=\beta_a,\qquad \widetilde{S}=S^{-1}.
$$
We can always choose an element $e_F\in\{e_A \, | \, |A|=0\mod 2\}\in\cl_{p,q}$ such that
\begin{eqnarray}
M_F:=\beta_A e_F e^A \neq 0.\label{MMFG}
\end{eqnarray}
Next, we can find the rotors
\begin{eqnarray}
S=\pm \frac{M_F}{\sqrt{\widetilde{M_F}M_F}}.\label{SSFG}
\end{eqnarray}
\end{corollary} 

In the particular case of dimension $n=3$, the formula for $M_F$ (\ref{MMF}) can be simplified. We obtain the following theorem.

\begin{theorem}[$n=3$]\label{th34}
Let $P\in\SO_+(p,q)$, $n=p+q=3$. We can always choose a basis element $e_F\in\{e_D\, | \, |D|=0\mod 2\}\in\cl_{p,q}$ such that
\begin{eqnarray}
L_F:=e_F+ p_a^b e_b e_F e^a\neq 0.\label{MMM}
\end{eqnarray}
Next, we can find elements $S=\pm T\in\Spin_+(p,q)$ that correspond to $P=(p_a^b)\in\SO_+(p,q)$ as the two-sheeted covering $Se_aS^{-1}=p_a^b e_b$ in the following way:
\begin{eqnarray}
S=\pm \frac{L_F}{\sqrt{\widetilde{L_F}L_F}}.\label{SSS}
\end{eqnarray}
\end{theorem}
\begin{proof} In the case $n=3$, from (\ref{MMF}), we get
\begin{eqnarray}
M_F&=&p^B_A e_B e_F e^A=e_F+p_a^b e_be_F e^a+p_{a_1 a_2}^{b_1 b_2}  e_{b_1 b_2}e_F e^{a_1 a_2}+p_{123}^{123}e_{123}e_F e^{123}\\
&=&2(e_F+p_a^b e_b e_F e^a),
\end{eqnarray}
where we used $e_{123}\in\cen(\cl_{p,q})$, $p^{12}_{12}=\det(P)=1$, $p^1_1=p^{23}_{23}$, $p^2_2=p^{13}_{13}$, and $p^3_3=p^{12}_{12}$. Because of (\ref{SSF}), we can normalize $M_F$ and obtain the statement of the theorem.
\end{proof}

Note that in the particular case $F=\emptyset$, the statement of Theorem \ref{th34} is presented in \cite{Doran} with a different proof. This reference does not provide any information about what we can do if $L_\emptyset=0$. Theorem \ref{th34} gives the answer for the general case.

Note that $\SO_+(p,q)\simeq \SO_+(q,p)$ and $\Spin_+(p,q)\simeq \Spin_+(q,p)$. Below we consider only the cases $p\geq q$.

\begin{example}
In the case $n=1$, we have 
$$
\SO(1)=\{ 1\},\qquad \Spin(1)=\{ \pm e\}.
$$
For the element $P=1\in\SO(1)$, there are two elements $S=\pm T=\pm e\in\Spin(1)$. In Theorem \ref{thnew}, we have $F=\emptyset$ and $M_\emptyset=e$ for this case.
\end{example}

\begin{example}\label{ex5} In the case $n=p=2$ and $q=0$, we have
\begin{eqnarray}
\SO(2)&=&\{
\begin{bmatrix}
\cos \phi \,&\, -\sin\phi \\
\sin\phi \,&\, \cos\phi \\
\end{bmatrix}\,|\, \phi\in\R 
\},\\
\Spin(2)&=&\{ ae+be_{12}\in\cl_{2} \,|\,  a^2+b^2=1,\, a, b\in\R\}.
\end{eqnarray}
From (\ref{MMF}), we get (with substitution $p^{12}_{12}=\det(P)=1$)
\begin{eqnarray}
\!\!\!\!\!\!M_F\!=\!p^B_Ae_B e_Fe^A\!=\!e_F+p^1_1e_Fe^1+p^1_2 e_1 e_Fe^2+p^2_1 e_2 e_Fe^1+p^2_2 e_2 e_F e^2+e_{12}e_Fe^{12}.\label{qwe}
\end{eqnarray}
Choosing $F=\emptyset$, we get
\begin{eqnarray}
M_\emptyset&=&(2+p^1_1+p^2_2)e+(p^1_2-p^2_1)e_{12}=2(1+\cos\phi)e-2\sin\phi\, e_{12},\\
\widetilde{M_\emptyset}M_\emptyset&=&(2(1+\cos\phi)e+2\sin\phi\, e_{12})(2(1+\cos\phi)e-2\sin\phi\, e_{12})\\
&=&8(1+\cos\phi)e=16 \cos^2 \frac{\phi}{2} \,e.
\end{eqnarray}
For the case $\phi\neq \pi+2\pi k, k\in\Z$, we get
\begin{eqnarray}
S=\pm \frac{M_\emptyset}{\sqrt{\widetilde{M_\emptyset} M_\emptyset}}=\pm \frac{4 \cos^2 \frac{\phi}{2}\, e-4\sin\frac{\phi}{2} \cos\frac{\phi}{2}\, e_{12}}{4\cos \frac{\phi}{2}}=\pm (\cos\frac{\phi}{2}\, e-\sin\frac{\phi}{2}\, e_{12}).\label{qw}
\end{eqnarray}
Choosing $F=12$, we get
\begin{eqnarray}
M_{12}&=&(p^1_2-p^2_1)e+(2-p^1_1-p^2_2)e_{12}=-2\sin\phi\,e+2(1-\cos\phi)e_{12},\\
\widetilde{M_{12}}M_{12}&=&(-2\sin\phi\,e-2(1-\cos\phi) e_{12})(-2\sin\phi\,e+2(1-\cos\phi) e_{12})\\
&=&8(1-\cos\phi)e=16 \sin^2 \frac{\phi}{2} \,e.
\end{eqnarray}
For the case $\phi\neq 2\pi k, k\in\Z$, we obtain the same result as in (\ref{qw}):
\begin{eqnarray}
\!\!\!S=\pm \frac{M_{12}}{\sqrt{\widetilde{M_{12}} M_{12}}}=\pm \frac{4\sin\frac{\phi}{2} \cos\frac{\phi}{2}\,e+4 \sin^2 \frac{\phi}{2}\, e_{12}}{4\cos \frac{\phi}{2}}=\pm (-\cos\frac{\phi}{2}\, e+\sin\frac{\phi}{2}\, e_{12}).
\end{eqnarray}
Note that we can also check directly for arbitrary $\phi\in\R$:
\begin{eqnarray}
&&S(xe_1+ye_2)S^{-1}=((\cos\frac{\phi}{2}) e-(\sin\frac{\phi}{2}) e_{12})(xe_1+ye_2)((\cos\frac{\phi}{2})+(\sin\frac{\phi}{2}) e_{12})\nonumber\\
&&=e_1((\cos \phi) x+(\sin \phi) y)+e_2 ((\cos\phi) y-(\sin \phi) x),\qquad x,y\in\R,
\end{eqnarray} 
and get the same result
\begin{eqnarray}
Se_a S^{-1}=p_a^b e_b,\qquad S=\pm T= \pm(\cos\frac{\phi}{2}\,e-\sin\frac{\phi}{2}\,e_{12}).
\end{eqnarray}
\end{example}

\begin{example}\label{ex6} In the case $p=q=1$, we have
\begin{eqnarray}
\SO_+(1,1)&=&\{
\begin{bmatrix}
\cosh \phi \,&\, \sinh\phi \\
\sinh\phi \,&\, \cosh\phi \\
\end{bmatrix}\,|\, \phi\in\R 
\},\\
\Spin_+(1, 1)&=&\{ ae+be_{12}\in\cl_{1,1} \,|\,  a^2-b^2=1,\, a, b\in\R\}.
\end{eqnarray}
From (\ref{MMF}), we again obtain (\ref{qwe}). Choosing $F=\emptyset$, we get
\begin{eqnarray}
M_\emptyset&=&(2+p^1_1+p^2_2)e-(p^1_2+p^2_1)e_{12}=2(1+\cosh\phi)e-2\sinh\phi\, e_{12},\\
\widetilde{M_\emptyset}M_\emptyset&=&(2(1+\cosh\phi)e+2\sinh\phi\, e_{12})(2(1+\cosh\phi)e-2\sinh\phi\, e_{12})\\
&=&8(1+\cosh\phi)e=16 \cosh^2 \frac{\phi}{2} \,e.
\end{eqnarray}
For arbitrary $\phi\in\R$, we have $M_\emptyset\neq 0$ and get
\begin{eqnarray}
\!\!\!\!S\!=\!\pm \frac{M_\emptyset}{\sqrt{\widetilde{M_\emptyset} M_\emptyset}}\!=\!\pm \frac{4 \cosh^2 \frac{\phi}{2}\, e-4\sinh\frac{\phi}{2} \cosh\frac{\phi}{2}\, e_{12}}{4\cosh \frac{\phi}{2}}\!=\!\pm (\cosh\frac{\phi}{2}\, e-\sinh\frac{\phi}{2}\, e_{12}).\label{qww}
\end{eqnarray}
Note that we can also check directly for arbitrary $\phi\in\R$:
\begin{eqnarray}
&&\!\!\!S(xe_1+ye_2)S^{-1}=((\cosh\frac{\phi}{2}) e-(\sinh\frac{\phi}{2}) e_{12})(xe_1+ye_2)((\cosh\frac{\phi}{2}) e+(\sinh \frac{\phi}{2}) e_{12})\nonumber\\
&&\!\!\!=e_1((\cosh \phi) x+(\sinh \phi) y)+e_2 ((\sinh\phi) x+(\cosh \phi) y,\qquad x, y\in\R,
\end{eqnarray}
and get the same result
\begin{eqnarray}
Se_a S^{-1}=p_a^b e_b,\qquad S=\pm T= \pm(\cosh(\frac{\phi}{2})e-\sinh(\frac{\phi}{2})e_{12}).
\end{eqnarray}
\end{example}

\begin{example}\label{ex7} In the case $n=p=3$ and $q=0$, we have 
\begin{eqnarray}
\!\!\!\!\!\!\!\!\!\SO(3)&=&\{P\in\Mat(3,\R) \,| \, P^\T T=I_3,\quad \det(P)=1\},\\
\!\!\!\!\!\!\!\!\!\!\Spin(3)&=&\{ ae+be_{12}+c e_{13}+d e_{23}\in\cl_{3}\,|\, a^2+b^2+c^2+d^2=1,\, a, b, c, d\in\R\}.
\end{eqnarray}
From (\ref{MMM}), we get for $F=\emptyset, 12, 13$, and $23$ respectively:
\begin{eqnarray}
L_\emptyset&=&(1+p^1_1+p^2_2+p^3_3)e+(p^1_2-p^2_1)e_{12}+(p^1_3-p^3_1)e_{13}+(p^2_3-p^3_2)e_{23},\nonumber\\
L_{12}&=&(p^1_2-p^2_1)e+(1-p^1_1-p^2_2+p^3_3)e_{12}+(-p^2_3-p^3_2)e_{13}+(p^1_3+p^3_1)e_{23},\label{fr}\\
L_{13}&=&(p^1_3-p^3_1)e+(-p^2_3-p^3_2)e_{12}+(1-p^1_1+p^2_2-p^3_3)e_{13}+(-p^1_2-p^2_1)e_{23},\nonumber\\
L_{23}&=&(p^2_3-p^3_2)e+(p^1_3+p^3_1)e_{12}+(-p^1_2-p^2_1)e_{13}+(1+p^1_1-p^2_2-p^3_3)e_{23}.\nonumber
\end{eqnarray}
Note that, for example, $L_\emptyset=L_{12}=0$ if and only if the matrix $P\in\SO(3)$ has the following form
$$\begin{bmatrix}
\cos \phi \,&\, \sin\phi \,&\, 0\\
\sin\phi \,&\, -\cos\phi \,&\, 0\\
0\,&\,0\,&\,-1
\end{bmatrix},\qquad \phi\in\R.
$$
For the matrices of this type, we have
$$
L_{13}=2(1-\cos\phi)e_1-2\sin\phi\, e_{23},\qquad L_{23}=-2\sin\phi\, e_{13}+2(1+\cos\phi)e_{23},
$$
which are not equal to zero at the same time (the first one is equal to zero for $\phi=2\pi k$, $k\in\Z$; the second one is equal to zero for $\phi=\pi+2\pi k$, $k\in\Z$). So, for example, for the matrix $P=\diag(1, -1, -1)\in\SO(3)$, we have $L_\emptyset=L_{12}=L_{13}=0$ and we need the last one $L_{23}=4e_{23}\neq 0$ for calculations using formula~(\ref{SSS}).
\end{example}

\begin{example}\label{ex8} In the case $p=2$ and $q=1$, we have 
\begin{eqnarray}
\SO_+(2,1)&=&\{P\in\Mat(3,\R)\, | \, P^\T \eta P=\eta,\, \det(P)=1,\, p^{12}_{12}\geq 1\},\\
\Spin_+(2,1)&=&\{ ae+be_{12}+c e_{13}+d e_{23}\in\cl_{2,1}\,|\, a^2+b^2-c^2-d^2=1,\, a, b, c, d\in\R\},\nonumber
\end{eqnarray}
where $\eta=\diag(1, 1, -1)$. From (\ref{MMM}), we get for $F=\emptyset, 12, 13$, and $23$ respectively:
\begin{eqnarray}
L_\emptyset&=&(1+p^1_1+p^2_2+p^3_3)e+(p^1_2-p^2_1)e_{12}+(-p^1_3-p^3_1)e_{13}+(-p^2_3-p^3_2)e_{23},\nonumber\\
L_{12}&=&(p^1_2-p^2_1)e+(1-p^1_1-p^2_2+p^3_3)e_{12}+(p^2_3-p^3_2)e_{13}+(-p^1_3+p^3_1)e_{23},\label{fr2}\\
L_{13}&=&(p^1_3+p^3_1)e+(-p^2_3+p^3_2)e_{12}+(1-p^1_1+p^2_2-p^3_3)e_{13}+(-p^1_2-p^2_1)e_{23},\nonumber\\
L_{23}&=&(p^2_3+p^3_2)e+(p^1_3-p^3_1)e_{12}+(-p^1_2-p^2_1)e_{13}+(1+p^1_1-p^2_2-p^3_3)e_{23}.\nonumber
\end{eqnarray}
These four expression are not equal to zero at the same time. Next, we use the formula (\ref{SSS}).
\end{example}

\section{Formalism of Matrices, Quaternions, and Split-quaternions}\label{sectionMatrix}

It is known that the spin groups $\Spin_+(p,q)$ are isomorphic to classical matrix Lie groups in the cases of small dimensions $n=p+q\leq 6$ (see, for example, \cite{Lounesto}). An important problem arises: how to reformulate the known methods for calculating elements of spin groups without using the Clifford algebra formalism. We give a complete answer to this problem in this section for the cases $n\leq 3$.

In this section, we use the well-known Pauli matrices
\begin{eqnarray}
\sigma_0=\begin{bmatrix} 1\,&\, 0\\ 0\,&\,1 \end{bmatrix},\qquad \sigma_1=\begin{bmatrix} 0\,&\, 1\\ 1\,&\, 0 \end{bmatrix},\qquad \sigma_2=\begin{bmatrix} 0\,&\, -i\\ i\,&\, 0 \end{bmatrix},\qquad \sigma_3=\begin{bmatrix} 1\,&\, 0\\ 0\,&\, -1 \end{bmatrix}\label{Pauli}
\end{eqnarray}
with the properties
$$
(\sigma_1)^2=(\sigma_2)^2=(\sigma_3)^2=-i\sigma_1\sigma_2\sigma_3=\sigma_0,\qquad \sigma_a \sigma_b=-\sigma_b \sigma_a,\, a\neq b,\, a, b=1, 2, 3.
$$

\begin{example}
In the case $n=1$, we have 
$$
\SO(1)=\{ 1\},\qquad \Spin(1)=\{ \pm e\}\simeq \OO(1).
$$
For the element $P=1\in\SO(1)$, there are two elements $X=\pm Y=\pm 1\in\OO(1)$.
\end{example}

\begin{example}
For the matrix $P=\begin{bmatrix}
\cos \phi \,&\, -\sin\phi \\
\sin\phi \,&\, \cos\phi \\
\end{bmatrix}\in\SO(2)$, we have two elements 
$$X=\pm Y=\pm (\cos(\frac{\phi}{2})+i\sin(\frac{\phi}{2}))=\pm \exp({\frac{i\phi}{2}})\in\U(1)$$
with respect to the two-sheeted covering of $\SO(2)$ by $\U(1)$. 

See the calculations in Example \ref{ex5}. We have
$$
\Spin(2)=\{ ae+be_{12}\in\cl_{2} \,|\,  a^2+b^2=1,\, a, b\in\R\}\simeq\U(1)=\{\exp({i \varphi})\,|\, \varphi\in\R\}
$$
because $e\equiv 1$, $(e_{12})^2=-e$ in $\cl_{2}$, and we can identify $e_{12}$ with the imaginary unit~$i$. To obtain the matrix form of the connection (\ref{conn}), we can use the matrix representation of $\cl_2$ with $e \to \sigma_0$, $e_1 \to \sigma_1$, $e_2 \to \sigma_2$, and $e_{12}\to \sigma_1 \sigma_2=i \sigma_3$:
\begin{eqnarray}
X \sigma_a X^{-1}=p^b_a \sigma_b,\qquad a, b=1, 2;\qquad X=\pm (\cos(\frac{\phi}{2})\sigma_0+i\sin(\frac{\phi}{2})\sigma_3).
\end{eqnarray}
\end{example}

\begin{example}
For the matrix $\begin{bmatrix}
\cosh \phi \,&\, \sinh\phi \\
\sinh\phi \,&\, \cosh\phi \\
\end{bmatrix}\in\SO_+(1,1)$, we have two elements 
$$X=\pm Y=\pm (\cosh(\frac{\phi}{2})+\sinh(\frac{\phi}{2}))\in\GL(1,\R)$$
with respect to the two-sheeted covering of $\SO_+(1,1)$ by $\GL(1,\R)$. 

See the calculations in Example \ref{ex6}. We have the isomorphism
\begin{eqnarray}
\Spin_+(1,1)\simeq\GL(1,\R)=\R\setminus\{0\}=\{a+b\,|\, a^2-b^2=1,\, a,b\in\R\}
\end{eqnarray}
because $(e_{12})^2=e$ in $\cl_{1,1}$. To obtain the matrix form of the connection (\ref{conn}), we can use the matrix representation of $\cl_{1,1}$ with $e \to \sigma_0$, $e_1 \to \sigma_1$, $e_2 \to i\sigma_2$, and $e_{12}\to i\sigma_1 \sigma_2=-\sigma_3$.
\end{example}

\begin{theorem}\label{th30} Let $P=(p_a^b)\in\SO(3)$. Then at least one of the following four quaternions $Q_F$, $F\in\{\emptyset, 12, 13, 23\}$, is non-zero:
\begin{eqnarray}
Q_\emptyset&:=&(1+p^1_1+p^2_2+p^3_3)+(p^1_2-p^2_1)i+(p^1_3-p^3_1)j+(-p^2_3+p^3_2)k,\nonumber\\
Q_{12}&:=&(p^1_2-p^2_1)+(1-p^1_1-p^2_2+p^3_3)i+(-p^2_3-p^3_2)j+(-p^1_3-p^3_1)k,\label{fr3}\\
Q_{13}&:=&(p^1_3-p^3_1)+(-p^2_3-p^3_2)i+(1-p^1_1+p^2_2-p^3_3)j+(p^1_2+p^2_1)k,\nonumber\\
Q_{23}&:=&(p^2_3-p^3_2)+(p^1_3+p^3_1)i+(-p^1_2-p^2_1)j+(-1-p^1_1+p^2_2+p^3_3)k.\nonumber
\end{eqnarray}
Using this nonzero $Q_F\neq 0$, we can find two unit quaternions $X=\pm Y\in\H^u\simeq\SU(2)\simeq\Spin(3)$, which correspond to $P\in\SO(3)$ as the two-sheeted covering in the following way:
\begin{eqnarray}
X=\pm \frac{Q_F}{\sqrt{\bar{Q}_F Q_F}}\in\H^u,
\end{eqnarray}
where $\bar{Q}_F$ is the usual conjugation of quaternions.
\end{theorem}
\begin{proof} 
Using the correspondence $e\to 1$, $e_{12} \to i$, $e_{13}\to j$, and $e_{23} \to -k$ (where $i^2=j^2=k^2=ijk=-1$), we obtain the isomorphism with the group of unit quaternions
$$
\Spin(3)\simeq \H^u:=\{q=a+bi+cj+dk\in\H\, | \, \bar{q}q:=a^2+b^2+c^2+d^2=1\},
$$
which is also isomorphic to the group
$$
\SU(2):=\{P\in\Mat(2,\C)\,|\, P^\HH P=I_2\},
$$
where $\HH$ is the Hermitian transpose, with the correspondence
$$a\sigma_0+bi\sigma_3+ci\sigma_2+di\sigma_1=\begin{bmatrix} a+bi\,&\, c+di\\ -c+di\,&\, a-bi \end{bmatrix}\in\SU(2)\to a+bi+cj+dk\in\H^u.
$$
Further see Theorem \ref{th34} and Example \ref{ex7}.
\end{proof}

\begin{theorem}\label{th21} Let $P=(p_a^b)\in\SO_+(2,1)$. Then at lest one of the following four split-quaternions $Q_F$, $F\in\{\emptyset, 12, 13, 23\}$ is non-zero:
\begin{eqnarray}
Q_\emptyset&:=&(1+p^1_1+p^2_2+p^3_3)+(p^1_2-p^2_1)i+(-p^1_3-p^3_1)j+(p^2_3+p^3_2)k,\nonumber\\
Q_{12}&:=&(p^1_2-p^2_1)+(1-p^1_1-p^2_2+p^3_3)i+(p^2_3-p^3_2)j+(p^1_3-p^3_1)k,\label{fr4}\\
Q_{13}&:=&(p^1_3+p^3_1)+(-p^2_3+p^3_2)i+(1-p^1_1+p^2_2-p^3_3)j+(p^1_2+p^2_1)k,\nonumber\\
Q_{23}&:=&(p^2_3+p^3_2)+(p^1_3-p^3_1)i+(-p^1_2-p^2_1)j+(-1-p^1_1+p^2_2+p^3_3)k.\nonumber
\end{eqnarray}
Using this nonzero $Q_F\neq 0$, we can find two unit split-quaternions $X=\pm Y\in\H^u_s\simeq \SU(1,1)\simeq\Spin_+(2,1)$, which correspond to $P\in\SO_+(2,1)$ as the two-sheeted covering in the following way:
\begin{eqnarray}
X=\pm \frac{Q_F}{\sqrt{\bar{Q}_FQ_F}}\in\H^u_s,
\end{eqnarray}
where $\bar{Q}_F$ is the usual conjugation of split-quaternions.
\end{theorem}
\begin{proof} 
Using the correspondence $e\to 1$, $e_{12} \to i$, $e_{13}\to j$, and $e_{23}\to -k$  (where $i^2=-j^2=-k^2=-ijk=-1$), we obtain the isomorphism with the group of unit split-quaternions (we define the split-quaternion conjugation as $\bar{q}:=a-bi-cj-dk$)
$$
\Spin_+(2,1)\simeq \H^u_s:=\{q=a+bi+cj+dk\in \H_s\, | \, \bar{q} q:=a^2+b^2-c^2-d^2=1\},
$$
which is also isomorphic to the group
$$
\SU(1,1):=\{P\in\Mat(2, \C)\, | \, P^\HH \eta P=\eta\},\qquad \eta=\diag(1,-1),
$$
with the correspondence
$$
 a\sigma_0+bi\sigma_3+c\sigma_1-d\sigma_2=\begin{bmatrix} a+bi\,&\, c+di\\ c-di\,&\, a-bi \end{bmatrix}\in\SU(1,1)\to a+bi+cj+dk\in \H^u_s.
$$
Further see Theorem \ref{th34} and Example \ref{ex8}.
\end{proof}

\section{Conclusions}\label{Conclusions}

In this paper, we present a method for calculation of spin groups elements for known elements of pseudo-orthogonal groups with respect to the corresponding two-sheeted coverings. The method works in the case of arbitrary dimension and signature using the Clifford algebra formalism. We present explicit formulas in the cases of all signatures $(p,q)$ of space of dimension $n=p+q\leq 3$ in the matrix, quaternion, and split-quaternion formalisms. These cases are the most important for applications. Theorems \ref{thnew}, \ref{th34}, \ref{th30}, and \ref{th21} are presented for the first time. In this paper, we stop at the case $n=3$ due to the 10-page limitation of the paper. The cases $n=4$, $5$, and $6$, which are also important for applications and can also be reformulated purely in the matrix formalism, will be considered in the next paper.

\section*{Acknowledgments}

The author is grateful to the organizers and participants of The 33rd/35th International Colloquium on Group Theoretical Methods in Physics (ICGTMP, Group33/35, Cotonou, Benin, July 15--19, 2024) for fruitful discussions. 

The article was prepared within the framework of the project “Mirror Laboratories” HSE University “Quaternions, geometric algebras and applications”.


\begin{thebibliography}{0}


\bibitem{Hestenes} D. Hestenes, {\it Space-Time Algebra} (Gordon and Breach, New York, 1966).

\bibitem{Doran} C. Doran, A. Lasenby, {\it Geometric Algebra for Physicists} (Cambridge Univ. Press, Cambridge, 2003).

\bibitem{Lounesto} P. Lounesto, {\it Clifford Algebras and Spinors} (Cambridge Univ. Press, Cambridge, 2001).

\bibitem{Shirokovspin} D. Shirokov, Calculation of elements of spin groups using generalized Pauli's theorem, {\it Advances in Applied Clifford Algebras}. {\bf 25}:1 (2015), 227--244.

\bibitem{Shirokovspin2} D. Shirokov, Calculation of elements of spin groups using method of averaging in Clifford’s geometric algebra, {\it Advances in Applied Clifford Algebras}. {\bf 29} (2019), 50.

\bibitem{Dev} D. Shirokov, Development of the method of averaging in Clifford geometric algebras, {\it Mathematics}. {\bf 11}:16 (2023), 3607.

\end{thebibliography}
\end{document}